\documentclass[sigconf]{acmart}
\copyrightyear{2026}
\acmYear{2026}
\setcopyright{cc}
\setcctype{by}
\acmConference[KDD '26]{Proceedings of the 32nd ACM SIGKDD Conference on Knowledge Discovery and Data Mining V.2}{August 09--13, 2026}{Jeju Island, Republic of Korea}
\acmBooktitle{Proceedings of the 32nd ACM SIGKDD Conference on Knowledge Discovery and Data Mining V.2 (KDD '26), August 09--13, 2026, Jeju Island, Republic of Korea}
\acmDOI{10.1145/3770855.3818953}
\acmISBN{979-8-4007-2259-2/2026/08}
\usepackage{multirow}
\usepackage{amsmath,amsthm}

\newtheorem{proposition}{Proposition}

\AtBeginDocument{%
  }

\begin{document}

\title{AS-Bridge: A Bidirectional Generative Framework Bridging Next-Generation Astronomical Surveys}


\author{Dichang Zhang}
\affiliation{%
  \institution{Stony Brook University}
  \city{Stony Brook, NY}
  \country{USA}
}
\email{diczhang@cs.stonybrook.edu}

\author{Yixuan Shao}
\affiliation{%
  \institution{Stony Brook University}
  \city{Stony Brook, NY}
  \country{USA}
}
\email{yixuan.shao@stonybrook.edu}

\author{Simon Birrer}
\affiliation{%
  \institution{Stony Brook University}
  \city{Stony Brook, NY}
  \country{USA}
}
\email{simon.birrer@stonybrook.edu}

\author{Dimitris Samaras}
\affiliation{%
  \institution{Stony Brook University}
  \city{Stony Brook, NY}
  \country{USA}
}
\email{samaras@cs.stonybrook.edu}

\renewcommand{\shortauthors}{Dichang Zhang, Yixuan Shao, Simon Birrer, and Dimitris Samaras}

\begin{abstract}
The upcoming decade of observational cosmology will be shaped by large sky surveys, such as the ground-based LSST at the Vera C. Rubin Observatory and the space-based Euclid mission. While they promise an unprecedented view of the Universe across depth, resolution, and wavelength, their differences in observational modality, sky coverage, point-spread function, and scanning cadence make joint analysis beneficial, but also challenging.
To facilitate joint analysis, we introduce A(stronomical)S(urvey)-Bridge, a bidirectional generative model that translates between ground- and space-based observations. AS-Bridge learns a diffusion model that employs a stochastic Brownian Bridge process between the LSST and Euclid observations. The two surveys have overlapping sky regions, where we can explicitly model the conditional probabilistic distribution between them. We show that this formulation enables new scientific capabilities beyond single-survey analysis, including faithful probabilistic predictions of missing survey observations and inter-survey detection of rare events. 
These results establish the feasibility of inter-survey generative modeling. AS-Bridge is therefore well-positioned to serve as a complementary component of future LSST–Euclid joint data pipelines, enhancing the scientific return once data from both surveys become available. Data and code are available at \href{https://github.com/ZHANG7DC/AS-Bridge}{https://github.com/ZHANG7DC/AS-Bridge}.

\end{abstract}



\maketitle
\begin{figure}[h]
    \centering
    \includegraphics[width=0.9\linewidth]{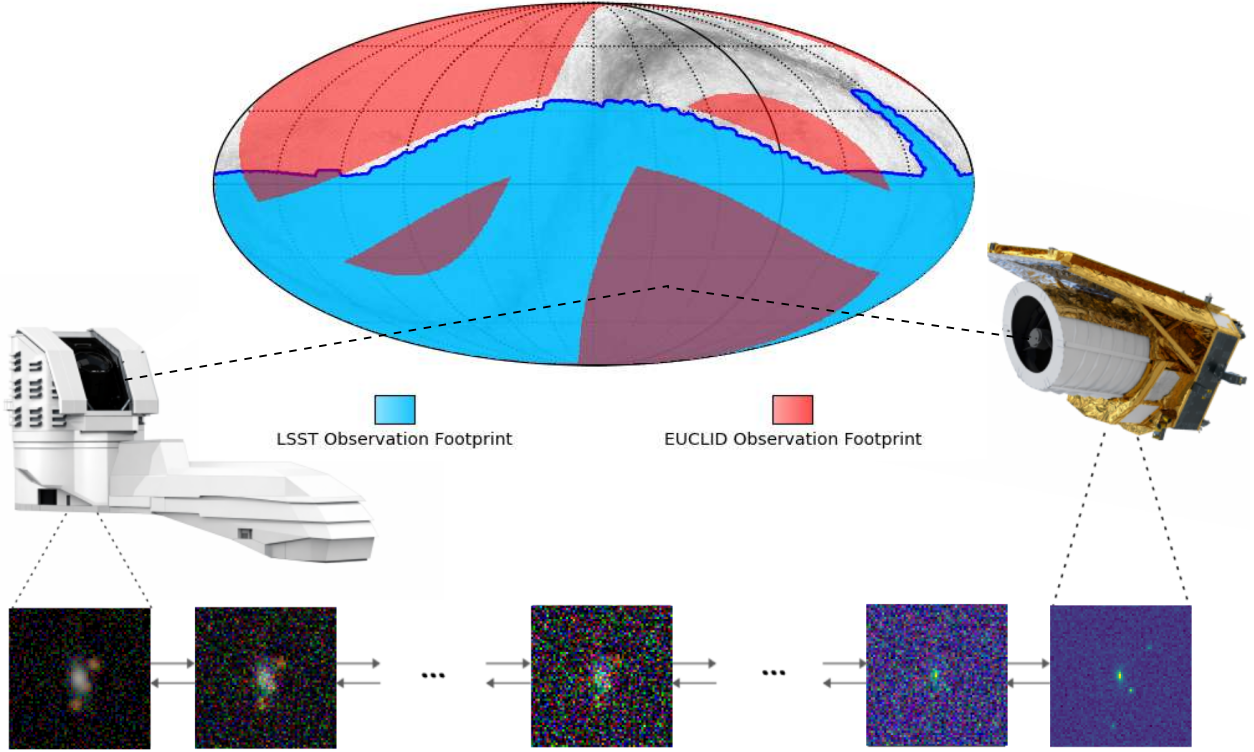}
    \caption{
    Overview of AS-Bridge. The central panel shows the visible sky, with the LSST and Euclid survey footprints marked in blue and red, respectively. The overlapping regions indicate areas jointly observed by the ground-based LSST (left) and the space-based Euclid mission (right). LSST provides multi-band optical images that are more blended due to atmospheric seeing, while Euclid delivers sharper near-infrared observations from space. From these overlapping regions, matched image cutouts are extracted and used to train AS-Bridge, which models the probabilistic translation between the two survey domains using a Brownian bridge formulation.
    }
    \Description{Diagram showing LSST and Euclid survey footprints, paired image cutouts, and the AS-Bridge Brownian bridge translation process between the two survey domains.}
    \label{fig:bidirectional_lsst_euclid_bridge}
\end{figure}
\section{Introduction}
In this decade, observational cosmology will be driven by major flagship surveys, such as the ground-based Legacy Survey of Space and Time (LSST) at the NSF–DOE Vera C. Rubin Observatory~\cite{ivezic2019lsst} and the space-based ESA Euclid mission~\cite{mellier2025euclid}. Together, they provide unprecedented coverage across depth, angular resolution, wavelength, and temporal cadence, enabling major advances in topics ranging from dark energy and dark matter to asteroids and other solar system objects~\cite{thelsstdarkenergysciencecollaboration2021lsstdarkenergyscience,2020,ivezic2019lsst,schwamb2018sssc,Wagg_2025,Kurlander_2025}.

Fundamental differences in observational modality—such as Point Spread Functions (PSFs), noise statistics, photometric bandpasses, and scanning cadence—introduce systematic distribution shifts that complicate cross-survey inference.

The mapping between ground-based and space-based imaging is inherently ill-posed in both directions. Recovering Euclid-quality morphology from LSST requires resolving ambiguity induced by atmospheric blurring and background noise. Conversely, mapping Euclid observations back to the LSST domain requires inferring spectral information from fewer and different bands, making the relationship non-identifiable from a single deterministic mapping. Therefore, inter-survey translation should be treated as a probabilistic process that can sample multiple valid realizations consistent with the available observations.

To model this probabilistic process, we introduce AS-Bridge, a unified bidirectional generative framework designed to model the stochastic connection between survey domains. AS-Bridge formulates inter-survey translation as a bidirectional Brownian bridge process, learning a stochastic path that connects the LSST and Euclid data distributions using overlapping sky regions as anchors. This formulation explicitly models the probabilistic transition between the two surveys and enables calibrated probabilistic inference in both directions.

We demonstrate the following benefits from inter-survey generation:

\textbf{Inter-survey restoration and synthesis.}
After training on paired cutouts—currently from simulations, but which can be readily constructed once official data releases become available after camera calibration—AS-Bridge can generate complementary observations in non-overlapping sky regions. We show that AS-Bridge recovers the target-survey observation conditioned on the source-survey input, and empirically demonstrate superior probabilistic reconstruction quality compared to competing baselines (Sec.~\ref{sec:ablation}). This enables analyses developed for one survey to be transferred to regions observed only by the other, while preserving uncertainty quantification.

\textbf{Inter-survey rare event detection.}
We cast rare-event detection as inter-survey inconsistency under uncertainty. At inference time, a paired observation is mapped into the inter-survey bridge and reconstructed back to one domain; the reconstruction inconsistency, aggregated over multiple stochastic samples, serves as an anomaly score. This yields an anomaly detector that simultaneously exploits complementary observational modalities through cross-survey conditioning and probabilistic reconstruction along the diffusion bridge. We demonstrate improved detection of diverse astronomical anomalies, using strong gravitational lenses as a representative example of structural outliers, compared with general-purpose baselines.

Our contributions are as follows:
\begin{itemize}
    \item We introduce the inter-survey modeling problem and formulate probabilistic inter-survey inference for next-generation cosmological surveys by treating cross-survey translation as an ill-posed, bidirectional inference problem. We further identify rare event detection as a key downstream application enabled by this formulation.

    \item We propose AS-Bridge, a generative framework that models inter-survey translation as a bidirectional Brownian bridge process, recovers the target conditional distribution given the source survey observation, and enables rare event detection through inter-survey probabilistic inconsistency.

    \item We construct the first realistic simulated benchmark for this problem and systematically compare AS-Bridge against representative baseline methods under standardized evaluation protocols and metrics.
\end{itemize}

\section{Background and Related Work}

\subsection{Next-Generation Astronomical Surveys}

This decade marks a paradigm shift in observational cosmology driven by the emergence of wide-field, high-throughput sky surveys. These next-generation surveys aim to map the Universe with unprecedented depth, angular resolution, spectral coverage, and temporal cadence, enabling precision studies of cosmology, galaxy evolution, and time-domain phenomena at scale.

Among them, the ground-based Legacy Survey of Space and Time (LSST)~\cite{ivezic2019lsst} at the Vera C. Rubin Observatory, the space-based Euclid mission, the Nancy Grace Roman Observatory~\cite{roman}, the Square Kilometer Array (SKA)~\cite{simon} and the Simon's Observatory. In this paper we focus on LSST and Euclid, as these are the first surveys on the sky with high complementary in resolution and temporal cadence.

LSST is a deep, multi-band optical imaging survey conducted by the Vera C. Rubin Observatory, located on Cerro Pachón in northern Chile. Its southern-hemisphere location enables continuous access to a large, contiguous region of the southern sky. Over a planned ten-year survey, LSST will repeatedly image approximately 18,000~deg$^2$ of sky across six optical bands ($u, g, r, i, z, y$), with a typical revisit cadence of a few days. The telescope features an effective aperture of 6.7~m and a 9.6~deg$^2$ field of view, producing images with a typical seeing-limited angular resolution of $\sim$0.7~arcsec. 

The strengths of LSST lie in its exceptional depth, temporal coverage, and statistical power, making it particularly well suited for studies of large-scale structure, weak gravitational lensing, and transient and variable phenomena. However, as a ground-based survey, LSST observations are affected by atmospheric turbulence, variable seeing conditions, sky background, and weather-dependent systematics.

In contrast, Euclid~\cite{mellier2025euclid} is a space-based mission operated by the European Space Agency, designed to deliver high-resolution imaging and near-infrared photometry and spectroscopy over a wide area of the extragalactic sky. Free from atmospheric effects, Euclid provides stable and sharply defined point-spread functions, enabling precise galaxy shape measurements. Euclid carries two primary instruments: the VIS instrument, which provides high-resolution optical imaging in a broad visible band with a pixel scale of 0.1~arcsec, and the NISP instrument, which performs near-infrared imaging and slitless spectroscopy across the $Y$, $J$, and $H$ bands.

Euclid is expected to survey approximately 15,000~deg$^2$ of the extragalactic sky, avoiding regions of high Galactic extinction and stellar density. Its primary science objectives focus on probing dark energy and dark matter through weak gravitational lensing and galaxy clustering, leveraging high-fidelity shape measurements and well-controlled instrumental systematics. However, compared to LSST, Euclid operates with fewer photometric bands and a different observational cadence, and its near-infrared focus results in complementary—but incomplete—spectral information relative to ground-based optical surveys.

The survey footprints of LSST and Euclid overlap significantly, with an estimated common area of $\sim$7,000--9,000~deg$^2$, depending on final survey strategies and masking. In these overlapping regions, both surveys observe the same underlying astrophysical objects but produce fundamentally different data realizations due to differences in angular resolution, wavelength coverage, noise statistics, and survey strategy. To date, most strong-lensing searches have been conducted within single-survey data sets, with training sets and detection pipelines tailored to the characteristics of that specific survey. These systematic differences induce non-trivial distribution shifts between the two datasets. As a result, direct joint analysis is challenging: models and measurements developed for one survey cannot be naively transferred to the other, and simple deterministic mappings fail to capture the intrinsic ambiguities and uncertainties between observational modalities.

\subsection{Strong Gravitational Lensing}\label{sec:strong-lensing}
Strong gravitational lensing occurs when a massive foreground galaxy (the deflector) lies close to the line of sight to a more distant background source, bending the source light strongly enough to produce multiple images and high-magnification features. For extended background galaxies, this typically appears as multiple arcs or a nearly complete Einstein ring. If the background galaxy also hosts a compact point source—such as an AGN, quasar, or supernova—the lens can additionally produce multiple point-like images, commonly observed as “doubles” (two images) or “quads” (four images). Systems dominated by extended lensed galaxies are often referred to as galaxy–galaxy lenses, while those highlighted by multiply imaged point sources are described using the double/quad terminology.

Compared to ordinary galaxy detections, strong lenses are much rarer because they require a near alignment between the foreground deflector and the background source, among other conditions. As a result, only a small fraction of galaxies on the sky exhibit recognizable strong-lensing morphologies in survey images.

In this work, we treat strong gravitational lenses as representative rare events. Although their physical nature is well understood, we intentionally treat them as unknown anomalies at test time, and evaluate whether the model can detect them purely through inter-survey probabilistic inconsistency.

\subsection{Image Translation}
Image-to-Image (I2I) translation studies how to transform an image from a \emph{source} visual domain $\mathcal{X}$ to a \emph{target} visual domain $\mathcal{Y}$. Given $x\in\mathcal{X}$, the goal is to synthesize $y\in\mathcal{Y}$ that retains the underlying semantics (e.g., geometry, layout, identity) but adopts domain-specific appearance statistics (e.g., texture, color distribution, sensor characteristics). 

Early practical I2I systems were largely driven by GANs, including pix2pix~\cite{pix2pix2017}, SPADE~\cite{park2019SPADE}, and OASIS~\cite{oasis}. 
More recently, diffusion models have emerged as a leading framework for image-to-image translation due to their stable training dynamics and likelihood-based generative formulation. Denoising Diffusion Probabilistic Models (DDPMs)~\cite{ddpm} define a forward process that gradually perturbs data with Gaussian noise and learn the corresponding reverse denoising dynamics for sampling. Conditional diffusion extends this framework to translation tasks by conditioning the reverse process on a source image, achieving strong performance across a wide range of I2I applications. Palette~\cite{saharia2022palette} provides a unified conditional diffusion framework for multiple restoration-style I2I problems.

In conditional diffusion pipelines such as Palette, target images are generated by reversing a diffusion process initialized from pure Gaussian noise, while the source image is incorporated solely as an external conditioning signal. 
An alternative formulation is provided by bridge-based diffusion models, which directly model a stochastic transition between two domain distributions. Brownian Bridge~\cite{li2023bbdm,LyuConsecBB} defines a forward process between source and target domains, rather than between data and pure noise. This formulation leads to more efficient sampling by avoiding unnecessary diffusion through high-noise states. In our work, we show that training Brownian bridges with an appropriate weighting term not only yields high-quality reconstructions but also fits better to the underlying conditional distribution.

Finally, although GAN-based, diffusion-based, and bridge-based methods can all generate multiple outputs for a single input, they are typically developed and evaluated under a single-direction point-estimation paradigm and evaluated on deterministic metrics. In contrast, scientific applications require models that faithfully represent the full conditional distribution between observational modalities. This motivates our focus on probabilistic inter-domain modeling and bidirectional inference beyond conventional I2I evaluation protocols.

\subsection{Anomaly Detection}
One of the high-impact opportunities of next-generation surveys is the discovery of ``unknown unknowns''—rare astrophysical events that challenge current astrophysical or cosmological models. This task can be characterized as Anomaly Detection (AD). Current approaches in astrophysics predominantly focus on temporal anomalies (transients)~\cite{10.1093/mnras/staf1477, DEBONIS2026101049, Ishida2021ActiveADTimeDomain, gómez2024machinelearningdrivenanomalydetection}. The standard pipeline involves extracting light-curve features followed by classical outlier detection algorithms like Isolation Forests~\cite{isolatoinforest}. While effective for time-domain events, these methods struggle to identify structural static anomalies—such as galaxy-galaxy strong gravitational lenses and other unique morphological irregularities—which require structural semantic understanding rather than temporal feature extraction.

Industrial visual anomaly detection (e.g., defect inspection) has recently converged on diffusion-based reconstruction~\cite{deco-diff,omiad}: a diffusion model is trained on normal data and anomalies are detected via elevated reconstruction/denoising error in regions that deviate from in-distribution examples. Similar reconstruction-based diffusion paradigms have also shown promise in medical imaging settings~\cite{10.1007/978-3-031-16452-1_4}.

More recently, multimodal anomaly detection benchmarks have been proposed to integrate complementary sensory inputs, such as combining 2D images with 3D point clouds~\cite{wang2023multimodal,costanzino2024cross}. These methods typically rely on explicit feature fusion modules to aggregate modality-specific representations and define anomaly scores based on cross-modal discrepancy.

\section{Method}
\subsection{Survey Translation Formulation}
\label{sec:problem_formulation}
Let the intrinsic light omission process underlying an astronomical object be denoted by $\Phi$.
In practice, this process is not observed directly.
Instead, observations are produced through the telescope and detector via a sequence of survey-specific operations, including convolution with the point-spread function (PSF), spectral integration over a limited set of wavelength bands determined by the instrument filters, and the addition of stochastic noise sources.

We model the observation process of a given survey as
\begin{equation}
x = \mathcal{O}(\Phi) + \epsilon\!\left(\mathcal{O}(\Phi)\right),
\end{equation}
where $\mathcal{O}$ denotes the deterministic observation operator incorporating PSF convolution and bandpass integration, and $\epsilon(\cdot)$ represents signal-dependent noise, including photon noise as well as signal-independent components such as readout and instrumental noise.

We consider two astronomical surveys observing the same underlying astrophysical scene but with different instrumental and observational characteristics.
Rather than treating one survey as a source domain and the other as a target domain, we view the observations as two distinct realizations of a shared latent astrophysical process, shaped by survey-specific effects such as atmosphere, optics, point-spread functions, noise statistics, and spectral bandpasses.
Formally, the two observations can be written as
\begin{equation}
\begin{aligned}
x_{\mathrm{Euclid}} &= \mathcal{O}_{\mathrm{Euclid}}(\Phi)
+ \epsilon_{\mathrm{Euclid}}\!\left(\mathcal{O}_{\mathrm{Euclid}}(\Phi)\right), \\
x_{\mathrm{LSST}}   &= \mathcal{O}_{\mathrm{LSST}}(\Phi)
+ \epsilon_{\mathrm{LSST}}\!\left(\mathcal{O}_{\mathrm{LSST}}(\Phi)\right).
\end{aligned}
\end{equation}

The latent process $\Phi$ is fundamentally unobservable: there exists no experiment, even in principle, that provides access to a noiseless, fully resolved, all-band representation of a distant astronomical object.
Therefore, we marginalize over the unobserved latent process $\Phi$ and focus on learning the relationship between different observational realizations.
Since each observation captures only a partial and noisy projection of the underlying scene, neither $x_{Euclid}$ nor $x_{LSST}$ fully determines the other.
The mapping between surveys is therefore inherently stochastic rather than deterministic.
We model the cross-survey relationship through the conditional distributions $p(x_{Euclid} \mid x_{LSST})$ and $p(x_{LSST} \mid x_{Euclid})$.

Our goal is to learn a single generative framework that jointly captures both conditionals, enabling bidirectional translation between surveys while preserving shared astrophysical structure and explicitly modeling the uncertainty induced by incomplete observations and signal-dependent noise.
Under this formulation, given an observation from either survey, the model can sample plausible realizations of how the same object would appear if observed by the other survey.

\subsection{AS-Bridge Framework}
\label{sec:bb_diffusion}

In this section, we first give a brief introduction to Brownian bridge diffusion and its training objectives used in previous work. We then derive a better training objective that balances between reconstruction fidelity and data likelihood which are equally important for probabilistic reconstruction problems. 

\subsubsection{Brownian Bridge}
\label{sec:brownian_bridge}

A Brownian bridge~\cite{oksendal2003stochastic} can be viewed as a Gaussian stochastic process obtained by conditioning Brownian motion on fixed endpoints.
In contrast to standard diffusion, which progressively destroys information from a single source, a Brownian bridge performs a stochastic interpolation between two observations.

Given two endpoints $(x_0, x_T)$, Brownian bridge motion between them is defined as

\begin{equation}
dx_t
=
\frac{x_0 - x_T}{1 - m_t}\,dt
+
dW_t,
\quad 0 \le t < T,
\end{equation}
where $W_t$ is the standard Brownian motion and $m_t = \frac{t}{T}$. The closed-form solution is

\begin{equation}
x_t
=
x_0
+
m_t(x_t - x_0)
+
\left(
W_t - m_t W_T
\right).
\end{equation}

The conditional distribution of $x_t$ given $x_0$ and $x_T$ is

\begin{equation}
x_t \mid (x_0, x_T)
\sim
\mathcal{N}
\left(
(1-m_t)x_0 + m_t x_T,\;
\delta_t I
\right),
\end{equation}
where $\delta_t = m_t(1-m_t)$.

The reverse transition is obtained from Bayes' theorem~\cite{ddpm, li2023bbdm}:
\begin{equation}
\begin{aligned}
p_\theta(x_{t-1}\mid x_t,x_1)
&=
\frac{q(x_t\mid x_{t-1},x_t)\,q(x_{t-1}\mid x_0,x_t)}
     {q(x_t\mid x_0,x_T)} \\
&=
\mathcal{N}\!\left(\tilde{\mu}_t,\ \delta_{t|t-1} I\right).
\end{aligned}
\end{equation}
where

\begin{equation}
\tilde{\mu}_t
=
a_{t}\,x_t
+
b_{t}\,x_T
-
c_{t}\Big(m_t (x_T-x_0) + \sqrt{\delta_t}\,\epsilon\Big),
\end{equation}
with coefficients
\begin{equation}
a_{t}
=
\frac{\delta_{t-1}}{\delta_t}\frac{1-m_t}{1-m_{t-1}}
+
\frac{\delta_{t|t-1}}{\delta_t}(1-m_t),
\end{equation}
\begin{equation}
b_{t}
=
m_{t-1}
-
m_t\frac{1-m_t}{1-m_{t-1}}\frac{\delta_{t-1}}{\delta_t},
\end{equation}
\begin{equation}
c_{t}
=
(1-m_t)\frac{\delta_{t|t-1}}{\delta_t},
\end{equation}
and variance schedule
\begin{equation}
\delta_{t|t-1}
=
\delta_t
-
\delta_{t-1}\frac{(1-m_t)^2}{(1-m_{t-1})^2}.
\end{equation}

Previous works using Brownian bridge~\cite{ddpm, li2023bbdm} optimize the following loss
\begin{equation}
\mathcal{L}_
=
\mathbb{E}_{x_0,x_T,t,\epsilon}
\left[
\left\|
\Big(m_t (x_T-x_0) + \sqrt{\delta_t}\,\epsilon\Big)
-
f_\theta(x_t,x_T,t)
\right\|_2^2
\right].
\end{equation}
which is the drift term plus the denoising term in a single timestep.
\subsubsection{Maximum Likelihood Training}

Song et al.~\cite{song2021maximum} showed that score-based diffusion models can be trained by maximum likelihood through a properly weighted denoising score matching objective. 
Under variance–exploding (VE) diffusion, the likelihood objective leads to a timestep-dependent weighting proportional to $\delta_t$. 
While the data is gradually blended so the total variance is preserved in standard DDPM~\cite{ho2020denoising}, VE diffusion defines the diffusion process where the data variance preserve so the total variance increases as the noise scale $\delta_t$ grows. Since the source and target observations are two different realizations of the same underlying astronomical object, we assume that the marginal data variance of the two domains is identical, which matches the VE assumption.

Both Song et al.~\cite{song2021maximum} and Ho et al.~\cite{ddpm} observed that introducing this weighting often results in slightly worse FID. However, Song et al. demonstrated that this weighting significantly improves the estimated data likelihood.

For scientific problems involving probabilistic reconstruction, our goal is to ensure that the learned reverse process matches the conditional probability distribution. Therefore, we prefer an objective that better aligns with maximum likelihood as well as perceptual quality.

However, directly using $\delta_t$ as the weighting term in Brownian bridge training leads to vanishing weights at both endpoints of the bridge. When optimizing for guaranteed maximum likelihood, such vanishing weights prevent the model from being effectively trained at the endpoints. 

To mitigate this issue, we propose to train the Browian bridge with an $\epsilon$-prediction loss. 
We formally prove that the $\epsilon$-prediction loss is  equivalent to the standard loss times a milder weight $\sqrt{\delta_t}$, which preserves the likelihood-inspired emphasis on high-noise timesteps while maintaining stable gradients across the bridge. 
Empirical results in Tab.~\ref{tab:energy_score} show that $\epsilon$-prediction yields the best reconstruction performance under uncertainty.
\begin{proposition}
Training with $\epsilon$-prediction loss is equivalent to training with the loss define in Eq. (12) with a milder weighting term $\sqrt{\delta_t}$.
\end{proposition}
\begin{proof}
From Eq.~(5), the forward process satisfies
\begin{equation}
x_t = \mu_t + \sqrt{\delta_t},\epsilon,
\end{equation}
where $\epsilon \sim \mathcal{N}(0, I)$.

The true score of $x_t$ is
\begin{equation}
\nabla_{x_t} \log p(x_t)
=
-\frac{x_t - \mu_t}{\delta_t}
=
-\frac{\epsilon}{\sqrt{\delta_t}}.
\end{equation}

If the network predicts noise $\epsilon_\theta(x_t, t)$, the implied score estimate is
\begin{equation}
s_\theta(x_t, t)
=
-\frac{\epsilon_\theta(x_t, t)}{\sqrt{\delta_t}}.
\end{equation}

Starting from the weighted denoising score matching loss
\begin{equation}
\mathcal{L}
=
\sqrt{\delta_t}
\left|\left|
s_\theta(x_t, t)
-
\nabla_{x_t} \log p(x_t)
\right|\right|^2_2,
\end{equation}
and substituting the expressions in terms of $\epsilon$, we obtain
\begin{equation}
\mathcal{L}
=
\sqrt{\delta_t}
\left|\left|
-\frac{\epsilon_\theta}{\sqrt{\delta_t}}
+
\frac{\epsilon}{\sqrt{\delta_t}}
\right|\right|_2^2
=
\left|\left|
\epsilon_\theta - \epsilon
\right|\right|_2^2.
\end{equation}

\end{proof}

With $\epsilon$-prediction, the estimated target at each timestep $\hat{x}_0$ can be derived based on Eq.(13)
\begin{equation}
    \hat{x}_0 = \frac{x_t-m_tx_T-\sqrt{\delta_t}\epsilon_\theta(x_t,x_T,t)}{1-m_t}
\end{equation}

\subsection{Rare event detection}
\label{sec:anomaly_detection}
We leverage the trained inter-survey generative model as an unsupervised anomaly detector.
The core assumption is that rare or previously unseen astrophysical phenomena are underrepresented in the training distribution and therefore cannot be faithfully reconstructed by the model.

Given paired observations $(x_{Euclid}, x_{LSST})$ from two surveys, we generate an intermediate variable $x_t$
 by stochastically fusing the two observations according to the Brownian bridge forward process defined in Eq.~(5). Starting from this noisy fusion, we progressively apply the learned reverse dynamics to reconstruct the original observations in each survey domain. Since the prediction of LSST observations is more uncertain than the prediction of Euclid, we use only the Euclid reconstruction error as the anomaly score.

Unlike previous reconstruction-based anomaly detection methods that directly compare a single reconstruction with the ground truth, astronomical images typically exhibit low signal-to-noise ratios, causing reconstruction errors to be dominated by mismatches in random noise rather than meaningful structural deviations.
To mitigate this issue, we sample multiple reconstructions from the same noisy fusion.
Given $N$ stochastic reconstructions $\{\hat{x}_0^{(i)}\}_{i=1}^N$, we define a pixel-wise anomaly map as
\begin{equation}
\mathcal{A}(p)
=
\min_{i \in \{1,\dots,N\}}
\left\|
\hat{x}_0^{(i)}(p) - x_0(p)
\right\|_2^2,
\end{equation}
and the image-level anomaly score is obtained by flux-normalized aggregation, avoiding bias toward brighter objects,
\begin{equation}
\mathcal{A}(x_0)
=
\frac{\sum_{p} \mathcal{A}(p)}{\sum_{p} x_0(p)}.
\end{equation}

By taking the minimum error across multiple stochastic reconstructions, this score suppresses spurious errors caused by noise fluctuations while remaining sensitive to systematic reconstruction failures.

\section{Experiments}
\subsection{Dataset Generation}
Euclid data products are becoming publicly available, while Rubin has revealed first imagery and early previews, with LSST science operations approaching but not yet represented by a fully mature public survey stream. To enable a controlled and survey-comparable development setting for cross-survey learning, we construct a forward-modeled dataset of galaxy observations under Euclid (VIS band) and LSST ($g,r,i$ bands) imaging. The dataset is built from astrophysically motivated population synthesis and rendered through survey-specific instrument transfer functions.

Using \textsc{SLSim}~\cite{slsim_github_2025}, we generate two primary image sets: regular galaxies and galaxy--galaxy strong-lens systems. Regular galaxies are drawn directly from the simulated galaxy catalogs (after basic selection cuts on, e.g., redshift and magnitude). Galaxy--galaxy lenses are produced by pairing foreground deflectors with background sources from the same catalogs and forward-modeling the resulting lensing configurations, retaining only systems that satisfy our detectability criteria (e.g., sufficient image separation and brightness). 

We additionally construct an auxiliary HST anomaly population for testing structural anomaly detection beyond strong lenses.

In total, we generate $115{,}000$ regular galaxies, $5{,}000$ galaxy--galaxy strong lenses, and 396 HST-template anomalies. We use $110{,}000$ regular galaxies for training, and reserve the remaining regular galaxies together with all lens and HST anomaly systems for evaluation and downstream tasks.

The simulation setting is deliberate: a mature public LSST survey
stream is not yet available, and confirmed labels for rare real events
are sparse and incomplete. Forward modeling provides paired obser-
vations and controlled ground truth for benchmarking inter-survey
inference. Once matched LSST–Euclid cutouts are available in over-
lapping sky regions, the same training recipe can be applied directly
to real observations

\subsubsection{Cosmology Configuration}
We assume a flat $\Lambda$CDM cosmology with $H_0 = 70~\mathrm{km\,s^{-1}\,Mpc^{-1}}$ and $\Omega_m = 0.3$. This fixes the angular-diameter distances that set the lensing geometry and the characteristic angular scale of Einstein radii.

\subsubsection{Regular Galaxy Population}\label{sec:galaxy}

We generate mock galaxy catalogs using \texttt{SkyPy}~\cite{Amara2021} over $z\in[0,5]$ with $\Delta z=0.1$ and limiting magnitude $m_{\mathrm{lim}}=30$ while sampling from an effective sky area of $0.1~\mathrm{deg}^2$.

We synthesize two physically interpretable galaxy populations: blue (late-type) and red (early-type) systems.

For each galaxy, we sample its redshift and absolute magnitude from redshift-evolving Schechter luminosity functions, assign a spectral energy distribution using Dirichlet mixtures of \texttt{kcorrect}-style templates, compute stellar mass and apparent magnitudes in Euclid and LSST bands, and draw structural properties (e.g., angular size and ellipticity) from population-dependent statistical priors.

This produces realistic regular galaxies with physically grounded photometry and morphology.
We instantiate two catalogs for different purposes: one catalog is used to render the regular galaxy samples for the dataset, while a second catalog is used as the parent population from which deflectors and sources are selected for constructing strong lens systems.

\subsubsection{Strong Lens Population}
Using the second catalog, we use the \textsc{SLSim} to construct galaxy--galaxy strong lens systems by explicitly separating \emph{deflectors} and \emph{sources} according to astrophysical priors. 

\textbf{Deflectors:} drawn from red and blue populations, filtered with $0.1 < z_d < 2$ and $i < 24$.

\textbf{Sources:} drawn from the blue population, filtered with $0.1 < z_s < 5$ and $i < 26$, and modeled as single Sérsic extended sources.

\textbf{Lens systems:} rendered over a nominal sky area of $10~\mathrm{deg}^2$, with image separation $0.8'' < \Delta\theta < 10''$, and a brightness cut applied to the second-brightest lensed image.

This produces a physically realistic distribution of galaxy--galaxy strong lenses consistent with cosmology and galaxy evolution.

\subsubsection{HST Anomaly Population}
To evaluate whether inter-survey inconsistency can detect structural anomalies beyond strong gravitational lenses, we construct an auxiliary HST-template anomaly population. We start from 396 anomalous HST cutouts spanning 15 morphological classes identified by an anomaly-search pipeline. Each HST anomaly is used as a high-resolution morphology template and forward-rendered through the same LSST and Euclid instrument transfer functions used for the regular and lens populations. We retain objects whose anomalous morphology remains visible after rendering in both survey domains. This yields paired LSST--Euclid anomaly examples with controlled cross-survey ground truth while preserving morphological diversity from real HST observations. These examples are used only for evaluation.

\subsubsection{Cross-survey Bandpass Harmonization}

To ensure photometric consistency across surveys, we harmonize bandpass definitions using \texttt{speclite}~\cite{speclite_zenodo_2024}. We adopt Euclid VIS and LSST 2016 bandpasses ($\text{lsst2016-g}$, $\text{lsst2016-r}$, $\text{lsst2016-i}$).\footnote{Because the Euclid VIS response contains extremely small non-zero values outside its effective wavelength range, which interferes with simulations over $z\in[0,5]$, we truncate the VIS filter to the official Euclid wavelength range of $5500$--$9000$ \AA. This avoids introducing non-physical flux contributions during SED integration.}

\subsubsection{Multi-survey Image Rendering (Euclid VIS, LSST $g,r,i$)}

Regular galaxies, lens systems, and HST anomaly templates are rendered into $64\times64$ pixel cutouts using survey-specific instrument models. \footnote{To avoid truncating extended lensing features in the lower-resolution LSST images (where the broader PSF and coarser sampling spread the signal over a larger apparent area), we choose the cutout field of view to be sufficiently large for LSST instead of the same FOV of Euclid.}

For each object, we render a Euclid VIS image and LSST $g$, $r$, and $i$ single-band images.

This produces paired cross-survey observations of the same underlying astrophysical object, observed through distinct instrument transfer functions (PSF, pixel scale, exposure, noise).

\subsection{Qualitative Evaluation and Potential Scientific Impact}
In this section, we qualitatively show the reconstruction examples of different start and end points of AS-Bridge and discuss how they may benefit downstream scientific analysis.
\subsubsection{LSST$\rightarrow$ Euclid} LSST observations are frequently affected by atmospheric seeing and PSF broadening, which cause nearby galaxies to merge into a single apparent source in crowded regions. In contrast, Euclid VIS provides significantly higher spatial resolution from space, clearly separating these systems and revealing their true multi-object structure. Galaxy demographics and weak lensing shape measurements rely critically on accurate measurements of galaxy counts, positions, and morphologies. When blended systems in LSST are misidentified as single objects, systematic biases are introduced into number counts and downstream population statistics. As shown in Fig.~\ref{fig:probabilistic_lsst_euclid} inferring from LSST to Euclid, AS-Bridge can separate blended systems and reconstruct a physically consistent multi-object scene from LSST observations, enabling more reliable demographic analysis on LSST-only sky region.

\subsubsection{Euclid $\rightarrow$ LSST}Reconstructing LSST multi-band color information from a single Euclid VIS observation is an extremely ill-posed problem, as Euclid collapses most spectral information into a broad band. However, Euclid still provides strong morphological cues and near-infrared intensity patterns that are statistically correlated with optical color. AS-Bridge learns this implicit conditional relationship between morphology, near-infrared brightness, and multi-band optical appearance across the two surveys.

As shown in Fig.~\ref{fig:probabilistic_euclid_lsst}, the model consistently preserves the morphological structure across multiple stochastic realizations while producing color variations that remain centered around the ground truth LSST observation. The diversity of these samples reflects the inherent uncertainty of the task, yet the ground truth lies well within the distribution of generated reconstructions. Tab.~\ref{tab:energy_score} quantitatively evaluates the quality of this reconstruction. This capability enables the analysis of plausible LSST band information for Euclid-only detections under uncertainty.

\subsubsection{LSST $\leftrightarrow$ Euclid} AS-Bridge learns to reconstruct typical in-distribution galaxy morphologies from the noisy fusion of LSST and Euclid observations. As shown in Fig.~\ref{fig:lens}, for common systems such as high-inclination disk galaxies, the model accurately recovers the elongated structure, demonstrating that it successfully captures the dominant galaxy population.

However, for rare galaxy--galaxy strong-lensing systems, the model exhibits a systematic reconstruction failure. This behavior arises because AS-Bridge is trained predominantly on regular galaxies and therefore learns the high-density modes of the regular-galaxy prior, such as centrally concentrated brightness profiles and smooth curvature patterns. Cross-survey translation is ill-posed, so the model balances consistency with the source observation against likelihood under the target-survey distribution. Rare lensing structures occupy low-density regions of this distribution and can be projected toward nearby high-likelihood regular morphologies during reconstruction.

In Fig.~\ref{fig:lens}, the model reconstructs the arc of a partially observable Einstein ring as a three point-source system. This discrepancy between the true Euclid observation and the AS-Bridge Euclid reconstruction provides a natural signal for anomaly and rare-event detection.
\begin{figure}[t]
    \centering
    \includegraphics[width=0.45\textwidth]{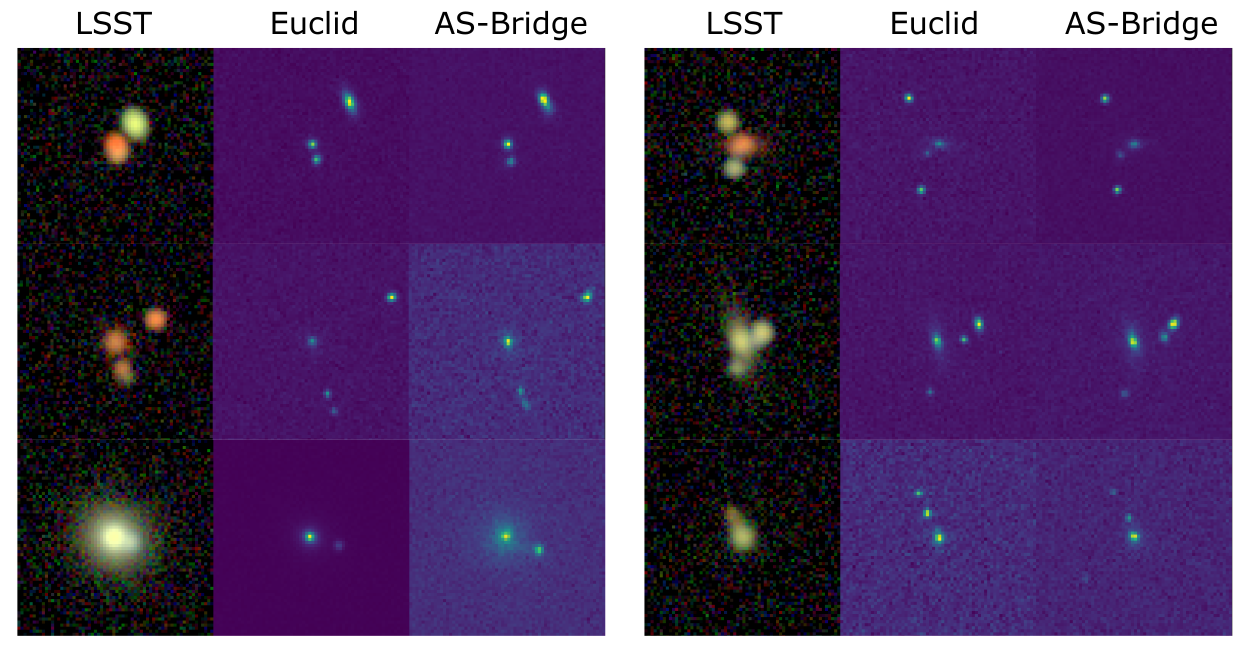}
    \vspace{-3mm}
    \caption{
    LSST-to-Euclid translation. Left: LSST where nearby galaxies are strongly blended due to atmospheric seeing and PSF broadening. Middle: Euclid observation revealing the true multi-object structure at higher spatial resolution. Right: AS-Bridge reconstruction from the LSST input, correctly recovering the number of galaxies and accurately localizing their centers.
    }
    \Description{Three-panel comparison of a blended LSST galaxy cutout, the corresponding sharper Euclid observation, and the AS-Bridge Euclid reconstruction.}
    \label{fig:probabilistic_lsst_euclid}
\end{figure}

\subsection{Quantitative Evaluation}
While the qualitative examples above illustrate the potential scientific impact of AS-Bridge, we next formalize this evaluation through quantitative metrics and comparisons with competing methods. These experiments establish a benchmark for future methods developed under our problem formulation and dataset.

\subsubsection{Metrics}  Evaluating AI models for scientific discovery requires metrics that address domain-specific constraints beyond standard computer vision tasks. First, telescope observation is inherently stochastic; thus, a model must not only minimize error but also correctly quantify observational uncertainty. Second, in the search for rare phenomena, the metric must reflect the practical workflow of domain experts, where the cost of verifying false positives is the limiting factor. Accordingly, we assess our framework using two complementary sets of metrics: Probabilistic Reconstruction Quality to validate physical calibration, and Detection Metrics to evaluate discovery efficiency.

To evaluate the probabilistic reconstruction quality of our Brownian Bridge model, we adopt the \emph{Continuous Ranked Probability Score} (CRPS)~\cite{Gneiting01032007}, a strictly proper scoring rule for probabilistic predictions.

\begin{figure}[t]
    \centering
    \includegraphics[width=0.45\textwidth]{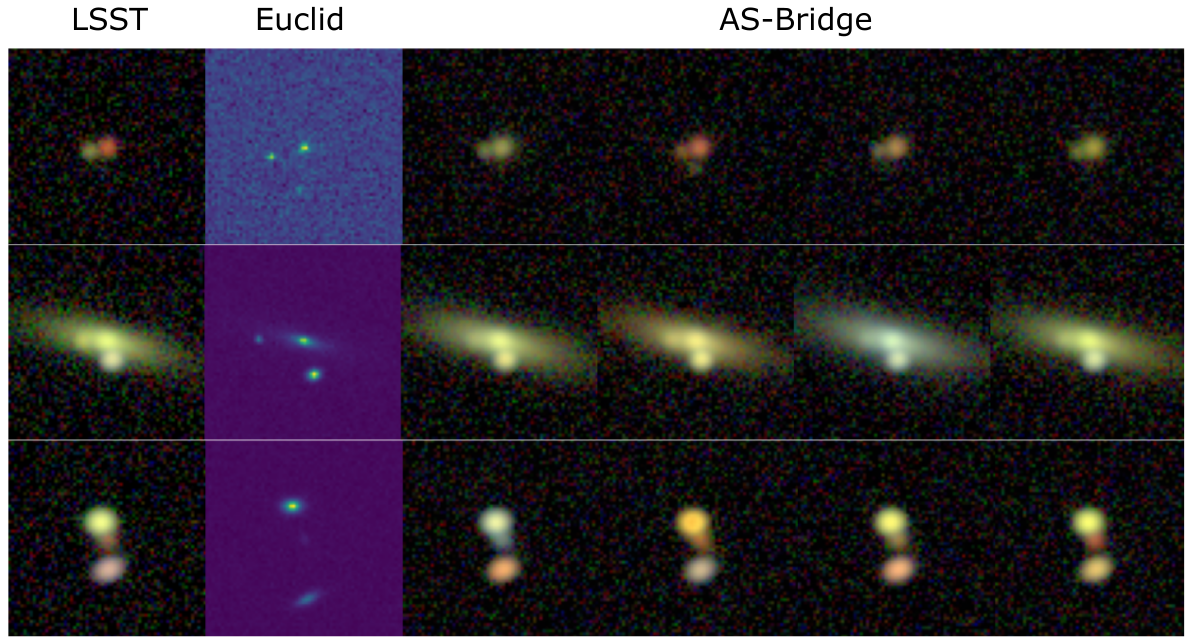}
   \vspace{-3mm}
    \caption{
   Euclid-to-LSST translation with multiple realizations. From left to right: LSST observation, Euclid  observation, and four stochastic LSST reconstructions generated by AS-Bridge from the Euclid input. While Euclid provides high-resolution morphology but limited spectral information, AS-Bridge reconstructs plausible LSST multi-band appearances that preserve the underlying structure and produce color variations consistent with the ground truth.
 }
    \Description{Row of image cutouts comparing a ground-truth LSST observation, the Euclid input, and multiple stochastic LSST reconstructions generated by AS-Bridge.}
    \label{fig:probabilistic_euclid_lsst}
\end{figure}

Given a ground-truth target $\mathbf{y} \in \mathbb{R}^d$ and $M$ samples $\{\hat{\mathbf{y}}^{(m)}\}_{m=1}^M$ drawn from the model’s predictive distribution, CRPS is defined as
\begin{equation}
\mathrm{CRPS} =
\frac{1}{M} \sum_{m=1}^{M} \left\| \hat{\mathbf{y}}^{(m)} - \mathbf{y} \right\|_2
\;-\;
\frac{1}{2M^2} \sum_{m=1}^{M} \sum_{m'=1}^{M}
\left\| \hat{\mathbf{y}}^{(m)} - \hat{\mathbf{y}}^{(m')} \right\|_2 .
\end{equation}

where $M$ is the total number of samples drawn from the stochastic prediction model and $\hat{y}^{(m)}$ is the $m$-th sample drawn. The first term measures accuracy with respect to the observed sample, while the second term penalizes lack of diversity among generated samples, thereby discouraging mode collapse. Because CRPS is strictly proper, it is minimized only when the predictive distribution matches the target distribution; overdispersed predictions are also penalized because excessive diversity increases the sample-to-truth distance in the first term. Rank-based calibration diagnostics are less direct for image-valued outputs because they require a non-unique projection of high-dimensional samples to an ordering, whereas CRPS evaluates the predictive distribution in the native image space.

CRPS has been widely adopted in scientific prediction tasks that require principled uncertainty quantification, including weather forecasting~\cite{rasp2023weatherbench} and probabilistic segmentation of ambiguous medical images~\cite{kohl2018probabilisticunet}.

\begin{figure}[t]
    \centering
    \includegraphics[width=0.45\textwidth]{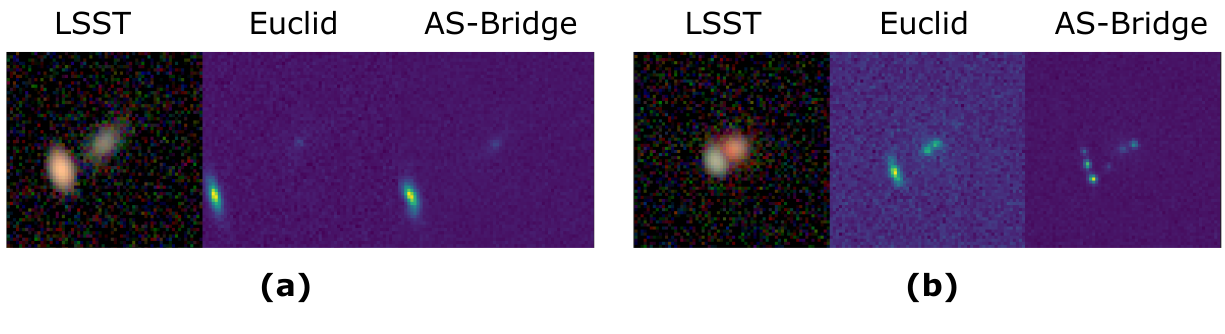}
    \vspace{-5mm}
    \caption{
    Midpoint-to-Euclid AS-Bridge reconstructions for (a) a system containing a common in-distribution edge-on disk galaxy and (b) a rare galaxy–galaxy strong-lensing system. From left to right: LSST , Euclid , and AS-Bridge reconstruction of Euclid from midpoint. Although the two cases appear visually similar, AS-Bridge faithfully reconstructs the elongated disk morphology in (a), while in (b) it fails to reproduce the extended arc. This facilitates rare event detection.
    }
    \Description{Two qualitative examples comparing LSST input, Euclid observation, and midpoint-to-Euclid AS-Bridge reconstructions for a regular edge-on galaxy and a rare strong-lensing system.}
    \label{fig:lens}
\end{figure}

In the context of astronomical anomaly detection—specifically the identification of rare phenomena  in large-scale surveys—the evaluation priorities differ significantly from standard industrial or safety-critical settings. 

In industrial defect detection, the goal is typically Recall-oriented (e.g., ensuring no defective part leaves a factory), prioritizing metrics such as the False Positive Rate (FPR) at high True Positive Rates (TPR) (e.g., 95\%). However, searching for rare scientific objects in massive datasets is Precision-oriented. The primary constraint is the limited ``time budget'' of domain experts who must verify candidates. A model that achieves high recall but floods the user with false positives is practically unusable in this domain.

To reflect this ``Needle-in-a-Haystack'' workflow, we adopt a set of discovery-focused metrics:

\begin{itemize}
    
    \item FPR@lowTPR: While industrial anomaly detection benchmarks often report the false positive rate (FPR) at 95\% true positive rate (TPR), such operating points are impractical for scientific discovery, where recovering nearly all anomalies would require scientists to sift through an overwhelming number of false positives. 
 Instead, we report the FPR at 1\% and 5\% TPR. This metric measures the fraction of regular objects that scientists would need to inspect in order to recover 1\% or 5\% of the true anomalies, reflecting a realistic early-discovery regime where verification effort must remain minimal.
    
    \item Area Under the Precision-Recall Curve (AUPR): Given the extreme class imbalance inherent to astronomical surveys, the standard AUROC can be misleadingly optimistic due to the massive number of True Negatives. We prioritize AUPR, which focuses exclusively on the positive class and penalizes false positives heavily, providing a more robust measure of global detection performance.
\end{itemize}

\begin{table*}[t]
    \centering
    \scriptsize
    \setlength{\tabcolsep}{2.2pt}
    \caption{Bidirectional translation metrics.}
    \label{tab:energy_score}
    \resizebox{\textwidth}{!}{%
    \begin{tabular}{lrrrrrrrrrr}
        \toprule
        \multirow{2}{*}{\textbf{Method}} &
        \multicolumn{5}{c}{\textbf{LSST $\to$ Euclid}} &
        \multicolumn{5}{c}{\textbf{Euclid $\to$ LSST}} \\
        \cmidrule(lr){2-6} \cmidrule(lr){7-11}
        & CRPS$\downarrow$ & PSNR$\uparrow$ & SSIM$\uparrow$ & MAE$\downarrow$ & RMSE$\downarrow$
        & CRPS$\downarrow$ & PSNR$\uparrow$ & SSIM$\uparrow$ & MAE$\downarrow$ & RMSE$\downarrow$ \\
        \midrule
        \multicolumn{11}{l}{\textit{Non-diffusion baselines}} \\
        SPADE~\cite{park2019SPADE} & 3.39 & 16.77 & 0.56 & 0.110 & 0.146 & 16.52 & 11.09 & 0.06 & 0.260 & 0.392 \\
        OASIS~\cite{oasis} & 4.65 & 16.78 & 0.57 & 0.110 & 0.146 & 13.33 & 11.34 & 0.10 & 0.250 & 0.352 \\
        pix2pix~\cite{pix2pix2017} & 4.35 & 16.95 & 0.60 & 0.110 & 0.143 & 73.03 & 12.81 & 0.07 & 0.210 & 0.289 \\
        \midrule
        \multicolumn{11}{l}{\textit{Diffusion baselines}} \\
        Palette~\cite{saharia2022palette} & 2.43 & 31.94 & 0.75 & 0.020 & 0.034 & 7.98 & 20.22 & 0.35 & 0.075 & 0.098 \\
        RePaint~\cite{9880056} & 3.14 & 11.23 & 0.22 & 0.220 & 0.289 & 15.15 & 10.84 & 0.05 & 0.250 & 0.290 \\
        \midrule
        \multicolumn{11}{l}{\textit{AS-Bridge w/ different training objectives}} \\
        $m_t(x_T - x_0) + \sqrt{\delta_t}\epsilon$
        & 2.55 & 32.59 & 0.76 & 0.019 & 0.033
        & \textbf{7.90} & 20.41 & 0.36 & \textbf{0.074} & \textbf{0.096} \\
        $\sqrt{\delta_t}\epsilon$
        & 3.59 & 31.38 & 0.75 & 0.026 & 0.046
        & 11.24 & 20.38 & 0.36 & 0.075 & 0.097 \\
        $\epsilon$
        & \textbf{2.38} & \textbf{32.90} & \textbf{0.77} & \textbf{0.018} & \textbf{0.031}
        & \textbf{7.90} & \textbf{20.42} & \textbf{0.37} & \textbf{0.074} & \textbf{0.096} \\
        \bottomrule
    \end{tabular}
    }
\end{table*}

\subsubsection{Competing Methods and Results}
\label{sec:ablation}
For the inter-survey translation task, we compare AS-Bridge of different training objectives with several other representative generative image translation methods. 
SPADE~\cite{park2019SPADE} is a spatially-adaptive normalization framework that injects semantic layout information into the generator. 
OASIS~\cite{oasis} is an adversarial semantic image synthesis method that replaces the discriminator with a segmentation network to stabilize training. 
pix2pix~\cite{pix2pix2017} is a conditional GAN for paired image-to-image translation. 
Palette~\cite{saharia2022palette} is a conditional diffusion model that uses the reference image as an additional conditioning signal.

In addition to these baselines, we adapt RePaint~\cite{9880056}, a diffusion inpainting method, to the inter-survey setting by treating the observed survey component as known and reconstructing the target survey component through reverse denoising.

Table~\ref{tab:energy_score} reports both CRPS and standard deterministic metrics for bidirectional translation. AS-Bridge with $\epsilon$-prediction achieves the best probabilistic reconstruction performance in both directions while remaining strongest or competitive under deterministic image metrics. Additional architecture, training, data-size, and sampling-cost details are provided in Appendix~\ref{app:implementation} and~\ref{app:ablations}.

For the anomaly detection task, we compare AS-Bridge with state-of-the-art single-modal and multi-modal industrial anomaly detection methods, namely Deco-Diff~\cite{deco-diff} and Crossmodal Feature Mapping (CFM)~\cite{costanzino2024cross}.

CFM is originally designed for cross-modal learning between 2D images and 3D point clouds using separate modality-specific encoders. In our setting, we adapt CFM by employing two image encoders to process LSST and Euclid observations respectively.

Table~\ref{tab:rare_event_detection} reports the performance of the three rare-event detectors on both strong-lensing and HST-anomaly test cases. Deco-Diff, as a single-modal method, fails to identify these events. In contrast, both AS-Bridge and CFM achieve strong performance, highlighting the importance of incorporating multi-sensor information for astronomical rare-event detection.
Furthermore, AS-Bridge outperforms CFM. This improvement is likely due to AS-Bridge’s stochastic reconstruction mechanism.

\begin{table}[t]
    \centering
    \scriptsize
    \caption{Rare-event detection across two anomaly test cases.}
    \label{tab:rare_event_detection}
    \resizebox{\columnwidth}{!}{%
    \begin{tabular}{llccc}
        \toprule
        \textbf{Test case} & \textbf{Method} & \textbf{FPR@1\% TPR} $\downarrow$ & \textbf{FPR@5\% TPR} $\downarrow$ & \textbf{AUPR} $\uparrow$ \\
        \midrule
        \multirow{3}{*}{Strong lensing} & AS-Bridge & \textbf{0.00\%} & \textbf{0.18\%} & \textbf{0.80} \\
        & Deco-Diff~\cite{deco-diff} & 1.1\% & 5.0\% & 0.61 \\
        & CFM~\cite{costanzino2024cross} & 0.24\% & 1.2\% & 0.75 \\
        \midrule
        \multirow{3}{*}{HST anomaly} & AS-Bridge & \textbf{0.00\%} & \textbf{0.00\%} & \textbf{1.000} \\
        & Deco-Diff~\cite{deco-diff} & 45.71\% & 63.13\% & 0.3297 \\
        & CFM~\cite{costanzino2024cross} & \textbf{0.00\%} & \textbf{0.00\%} & 0.962 \\
        \bottomrule
    \end{tabular}
    }
\end{table}

\section{Limitations and Ethical Considerations}
This work is currently trained and evaluated only on simulated data, and a simulation-to-reality domain gap is expected when transferring the model to real LSST and Euclid observations. It should therefore be regarded as a proof-of-concept. Demonstrating effectiveness in a controlled simulation setting allows us to justify and propose subsequent training on real survey data once they become available, rather than applying the method directly to raw observations without prior validation. This simulation-first validation is a common and well-established practice in astronomy.

This work involves only astronomical data and publicly available survey specifications. No human subjects, personal data, or sensitive information are involved. We do not identify any direct ethical risks associated with this research.

\section{Conclusion}
We formulate a probabilistic inter-survey translation problem between the two flagship astronomical surveys, LSST and Euclid, and propose AS-Bridge to model this stochastic relationship. We demonstrate several potential scientific impacts enabled by different choices of bridge start and end points, including deblending merged sources in crowded fields, synthesizing plausible multi-band color information with uncertainty, and detecting rare astronomical events. We expect additional scientific applications of this inter-survey bridge to emerge as the community explores its capabilities.

We further propose task-specific evaluation metrics and compare AS-Bridge against representative competing baselines. Together with the simulated dataset, this work establishes a benchmark for LSST–Euclid translation that can be used for future methodological development and evaluation.

As future work, we plan to deploy this framework on real observations once the official data releases from both surveys become available. We envision an iterative scientific pipeline built around the two core capabilities of AS-Bridge. A base AS-Bridge model, trained on the regular population, serves as a general-purpose translator and a tool for discovering rare events. These rare events, once inspected and categorized by scientists, can then be used to fine-tune category-specific AS-Bridge models, enabling more accurate reconstruction and analysis of objects belonging to particular rare-event classes.

\bibliographystyle{ACM-Reference-Format}
\bibliography{references}

\appendix
\section{Architecture and Training Details}
\label{app:implementation}
All diffusion baselines in our comparison, including AS-Bridge, Palette, and RePaint, use the same lightweight CNN U-Net backbone to isolate the effect of the bridge formulation and training objective. The backbone starts with a $3\times3$ convolution and uses three resolution stages with channel multipliers $(1,2,4)$ from a base width of 64 and two residual blocks per stage. Residual blocks use GroupNorm, SiLU activations, two $3\times3$ convolutions, and FiLM-style conditioning from a sinusoidal timestep embedding passed through a two-layer MLP with hidden dimension 256. The bottleneck has two additional residual blocks; the decoder uses nearest-neighbor upsampling followed by convolution and standard U-Net skip connections. AS-Bridge additionally uses a learned direction embedding to distinguish the two translation directions.

The diffusion models each use 12.03M parameters for both inference and training. For comparison, pix2pix uses 29.25M at inference and 32.01M during training, OASIS uses 68.45M and 90.70M, and SPADE uses 92.06M and 97.59M. Thus AS-Bridge's gains are not explained by larger model capacity. All methods use the same train/test split. Diffusion models use a linear schedule, AdamW with learning rate $10^{-4}$, batch size 128, and 300 epochs. AS-Bridge training takes 14.4 hours on one NVIDIA H100 GPU.

\begin{table}[H]
    \centering
    \scriptsize
    \caption{Ablation on training-set size.}
    \label{tab:appendix_ablation_train_size}
    \begin{tabular}{rcc}
        \toprule
        \textbf{Train samples} & \textbf{CRPS E$\to$L} & \textbf{CRPS L$\to$E} \\
        \midrule
        $1{,}000$ & 34.10 & 7.79 \\
        $10{,}000$ & 11.60 & 5.31 \\
        $20{,}000$ & 8.65 & 2.95 \\
        $40{,}000$ & 7.96 & 2.70 \\
        $60{,}000$ & 7.93 & 2.50 \\
        $80{,}000$ & 7.92 & 2.43 \\
        $110{,}000$ & 7.90 & 2.38 \\
        \bottomrule
    \end{tabular}
\end{table}

\begin{table}[t]
    \centering
    \scriptsize
    \caption{Ablation on reverse sampling steps. }
    \label{tab:appendix_ablation_sampling_steps}
    \begin{tabular}{rccc}
        \toprule
        \textbf{Steps} & \textbf{CRPS E$\to$L} & \textbf{CRPS L$\to$E} & \textbf{Samples/s} \\
        \midrule
        20 & 9.68 & 2.69 & 40.87 \\
        50 & 7.90 & 2.38 & 16.33 \\
        100 & 7.88 & 2.30 & 8.13 \\
        200 & 7.90 & 2.31 & 4.06 \\
        \bottomrule
    \end{tabular}
\end{table}

\section{Training-Data and Sampling Ablations}
\label{app:ablations}
Tables~\ref{tab:appendix_ablation_train_size} and~\ref{tab:appendix_ablation_sampling_steps} report how paired training-set size and reverse sampling steps affect AS-Bridge. Performance improves rapidly from $1{,}000$ to $40{,}000$ paired regular-galaxy cutouts and then begins to saturate. For sampling, 50 reverse steps provide a strong quality-throughput trade-off.

\end{document}